\documentclass[conference,letterpaper]{IEEEtran}

\usepackage[utf8]{inputenc} 
\usepackage[T1]{fontenc}
\usepackage{url}
\usepackage{ifthen}
\usepackage{cite}
\usepackage[cmex10]{amsmath} % Use the [cmex10] option to ensure complicance
\usepackage{eurosym}
\usepackage{color}
\usepackage{bbold}
\usepackage{graphicx}
\usepackage[parfill]{parskip} % begin paragraphs with an empty line rather than an indent
\usepackage{array}
\usepackage{amsmath}
\usepackage{subfigure}
\usepackage{stmaryrd} 
\usepackage{amsthm}

\newtheorem{prop}{Proposition}

\newcommand{\fig}[4]{ \begin{figure}[#4]
  \centering
   \includegraphics[width=#3\textwidth]{Figures/#1}
   \caption{#2}\label{fig:#1}
  \end{figure}
}

\begin{document}
\title{Practical Short-Length Coding Schemes for Binary Distributed Hypothesis Testing} 

% %%% Single author, or several authors with same affiliation:
% \author{%
%  \IEEEauthorblockN{Andrew R.~Barron}
%  \IEEEauthorblockA{Department of Statistics and Data Science\\
%                    Yale University\\
%                    New Haven, CT, USA\\
%                    Email: andrew.barron@yale.edu}
% }

%%% Several authors with up to three affiliations:
\author{
Elsa Dupraz$^{1}$, Ismaila Salihou Adamou$^{1}$, Reza Asvadi$^{2}$, and Tad Matsumoto$^{1,3}$ \\
\small $^{1}$ IMT Atlantique, CNRS UMR 6285, Lab-STICC, Brest, France \\  $^{2}$ Faculty of Electrical Engineering, Shahid Beheshti University, Tehran \\ $^{3}$ JAIST and University of
Oulu (Emeritus) \thanks{ This work has received a French government support granted to the Cominlabs excellence laboratory and managed by the National Research Agency in the ``Investing for the Future'' program under reference ANR-10-LABX-07-01.}
}

\maketitle

\begin{abstract}
 This paper investigates practical coding schemes for Distributed Hypothesis Testing (DHT).  While the literature has extensively analyzed the information-theoretic performance of DHT and established bounds on Type-II error exponents through quantize and quantize-binning achievability schemes, the practical implementation of DHT coding schemes has not yet been investigated. Therefore, this paper introduces practical implementations of quantizers and quantize-binning schemes for DHT, leveraging short-length binary linear block codes. Furthermore, it provides exact analytical expressions for Type-I and Type-II error probabilities associated with each proposed coding scheme. Numerical results show the accuracy of the proposed analytical error probability expressions, and enable to compare the performance of the proposed schemes. 
%   The paper concludes by numerically evaluating and comparing the performance of these schemes across various source and code parameters.
\end{abstract}

\section{Introduction}
In modern communication systems, the communication objective has evolved beyond the mere reconstruction of original information to encompass the application of Machine Learning tasks such as decision making or classification upon received data. This paper focuses on the specific case of Distributed Hypothesis Testing (DHT)~\cite{amari1998statistical}. In the DHT setup, the decoder aims to make a decision between two hypothesis $\mathcal{H}_0$ and $\mathcal{H}_1$, based on a coded version of the source $X$, and on some observed side information $Y$. The performance of a DHT scheme is typically evaluated from two error probabilities. Type-I error refers to deciding $\mathcal{H}_1$ while $\mathcal{H}_0$ was true, and Type-II error conversely involves selecting $\mathcal{H}_0$ when $\mathcal{H}_1$ is true. To fully address DHT, it is essential to investigate both information-theoretic performance limits and practical coding schemes for this setup. 

From an information-theoretic perspective, DHT was initially introduced by Berger in~\cite{berger1979decentralized}. Subsequently, Ahlswede and Csiszar established a lower bound on Type-II error-exponent through an achievability scheme based on a quantizer~\cite{ahlswede1986hypothesis}. 
Notably, this scheme was found to be optimal for testing against independence~\cite{watanabe2022sub}. 
Next, Shimokawa et al. improved the lower bound presented in~\cite{ahlswede1986hypothesis}, extending its applicability to general hypotheses, not restricted to testing against independence~\cite{shimokawa1994error}. This enhancement was achieved by employing a quantize-binning achievability scheme. 
The optimality of the quantize-binning scheme was then investigated on several DHT problems. It was shown in~\cite{rahman2012optimality} that this scheme is optimal for a certain class of testing against conditional independence problems, given that the tradeoff between quantization and binning is adressed properly. This analysis was further extended to encompass more generalized hypotheses, including those involving independent and identically distributed (i.i.d.) binary sources~\cite{haim2016binary}, non-binary sources~\cite{katz2015necessity}, and non-i.i.d. sources~\cite{adamou2023information}. The impact of discrete memoryless channels~\cite{sreekumar2019distributed}, Multiple-Access channels~\cite{salehkalaibar2018distributed}, or two-hop relay networks~\cite{salehkalaibar2019hypothesis} on the DHT performance was also analyzed. Multi-terminal source coding for DHT was also considered in~\cite{watanabe2017neyman}.  %This analysis was further extended to more general hypothesis with either i.i.d. binary~\cite{haim2016binary} or non-binary sources~\cite{katz2015necessity}, but also to non i.i.d. sources~\cite{adamou2023information}. 
%Finally, some other works considered more complex communication scenario and evaluated the impact of discrete memoryless channels~\cite{sreekumar2019distributed} or Multiple-Access channels~\cite{salehkalaibar2018distributed} on the error-exponent. 

While the DHT information-theoretic performance is now well known across a wide range of source models and hypotheses, the problem of designing practical coding schemes for this setup has been by far less investigated.  This paper takes a step towards bridging this gap, by focusing on the scenario where the sources $X$ and $Y$ are binary. %In this paper, we aim to address this problem at first for the simple case where the sources $X$ and $Y$ are binary. 

Previous achievability proofs suggest to consider either quantizer-alone or quantize-binning coding schemes for DHT. Some existing works have already introduced practical binary quantizers~\cite{fridrich2007binary}, binning schemes~\cite{xiong2004distributed}, and quantize-binning schemes~\cite{wainwright2009low}, all constructed with linear block codes. But the constructions of~\cite{fridrich2007binary,xiong2004distributed,wainwright2009low} primarily focus on source reconstruction and typically involve very long source sequences, often exceeding $10^5$ bits. However, DHT inherently deals with short-length sequences, where just a few dozen bits may suffice for making the correct decision. And the construction of efficient short-length linear block codes is often known to be a challenging problem~\cite{cocskun2019efficient}. Consequently, an important question arises regarding whether binary quantizers and quantize-binning schemes are efficient structures for practical short-length DHT. 

This paper addresses practical coding schemes for DHT by providing two main contributions. Firstly, it introduces and compares three practical coding schemes. The first scheme, named the truncation scheme, involves transmitting only the first $\ell$ bits of the source sequence $\mathbf{x}^n$ and serves as a performance baseline. The second and third schemes are practical quantizer-alone and quantize-binning schemes, both constructed with short-length binary linear block codes. %The 
Secondly, the paper provides exact analytical expressions for Type-I and Type-II error probabilities of the three considered coding schemes. These new analytical tools should enable the optimization and comparison of the proposed practical schemes across a wide range of source and code parameters. Numerical results validate the accuracy of the analytical error probabilities by demonstrating their consistency with Monte-Carlo simulations. Additionally, the paper compares the performance of the three schemes for sources of length $n=31$ bits. 

The outline of the paper is as follows. Section~\ref{sec:model} describes the binary DHT problem. Section~\ref{sec:truncated} presents the truncation scheme. Section~\ref{sec:quantizer} introduces the quantizer scheme. Section~\ref{sec:qandb} describes the quantize-binning scheme. Section~\ref{sec:simulation} shows and discusses numerical results. 

\section{Problem Statement}\label{sec:model}
In what follows, $\llbracket 1,M \rrbracket$ is the set of integers between $1$ and $M$. Random variables are represented by capital letters, such as $X$, while their realizations are denoted by lower-case letters, like $x$. Vectors of length $n$, such as $\mathbf{X}^n$, are in bold. In addition, we use $w(\mathbf{x}^n)$ to denote the Hamming weight of the vector $\mathbf{x}^n$, and $d(\mathbf{x}^n,\mathbf{y}^n)$ to denote the Hamming distance between $\mathbf{x}^n$ and $\mathbf{y}^n$.  The binomial coefficient of the pair of integers $(n, k)$ with $k \leq n$ is expressed as $\binom{n}{k}$.

\subsection{DHT for binary sources}

We consider binary source vectors $\mathbf{X}^n$ and $\mathbf{Y}^n$ both of length $n$. The encoder observes $\mathbf{X}^n$, whereas the decoder observes $\mathbf{Y}^n$, as illustrated in Fig.~\ref{fig:sw_ht}. The encoder sends a coded version of $\mathbf{X}^n$, while the decoder aims to make a decision between two hypothesis $\mathcal{H}_0$ and $\mathcal{H}_1$,  based on both $\mathbf{Y}^n$ and the received coded data.

We assume that the $n$ symbols of the sequences $\mathbf{X}^n$ and $\mathbf{Y}^n$ are i.i.d. and drawn according to random variables $X$ and $Y$, respectively.  Furthermore, $X$ and $Y$ are jointly distributed according to the model $Y = X \oplus E$, where $E$ is a binary random variable independent of $X$, and  $P(X=1) = 1/2$. Additionally, we denote $p = \mathbb{P}(E = 1)$ with $0<p \leq 1/2$. The two hypotheses are expressed as:
\begin{equation}\label{eq:dht_problem}
 \left\{
    \begin{array}{ll}
       \mathcal{H}_0 : & p = p_0, \\
        \mathcal{H}_1 : & p = p_1.
    \end{array}
\right.
\end{equation}
We assume, without loss of generality, that $p_0 < p_1$. It is worth noting that the probability distributions of both $X$ and $Y$ are independent of the chosen hypothesis, given $P(X=1) = 1/2$. 
This model was investigated from an information-theoretic perspective for instance in~\cite{shimokawa1994error,katz2015necessity,haim2016binary}. 
Furthermore, when $p=1/2$, the problem~\eqref{eq:dht_problem} reduces to testing against independence~\cite{ahlswede1986hypothesis}. 

\fig{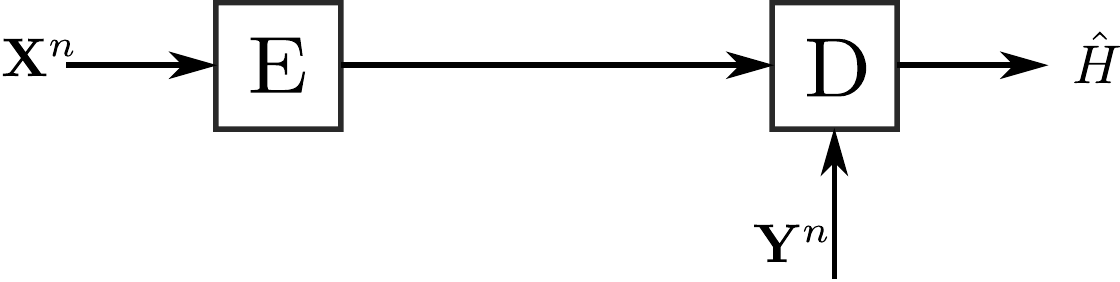}{Distributed hypothesis testing scheme}{0.4}{t}

\subsection{Error-exponent for binary DHT}\label{sec:error_exponent}
We will now restate an existing lower bound on the error exponent for the problem defined in Equation~\eqref{eq:dht_problem}, from the literature of information theory. Following the approach in~\cite{rahman2012optimality,shimokawa1994error}, we consider an encoding function:
\begin{equation}
 f_n: \{0,1\}^n \rightarrow \llbracket 1,2^v \rrbracket ,
\end{equation}
and a decision function $g_n:\llbracket 1,2^v \rrbracket\times \{0,1\}^n \rightarrow \{0,1\}$. We consider a rate-limited setup in which $v/n \leq R$. 

For given functions $(f_n,g_n)$, we define Type-I error probability $\alpha_n$ and Type-II error probability $\beta_n$ as in~\cite{haim2016binary} as
\begin{align}
 \alpha_n & = \mathbb{P}(g_n(  f_n(\mathbf{X}^n), \mathbf{Y}^n  ) = 1 | \mathcal{H}_0) , \\
 \beta_n & = \mathbb{P}(g_n(  f_n(\mathbf{X}^n), \mathbf{Y}^n  ) = 0 | \mathcal{H}_1) . 
\end{align}
For a given value $\epsilon \in (0,1)$ such that $\alpha_n < \epsilon$, Type-II error exponent $\theta$ is defined as~\cite{haim2016binary} 
\begin{equation}\label{eq:def_ee}
 \lim_{n \rightarrow \infty} \sup \frac{1}{n} \log_2 \frac{1}{\beta_n} \geq \theta.
\end{equation}
%where $\beta_n^{\star}$ is the minimum of $\beta_n$. 

A generic lower bound on the error exponent $\theta$ is provided in~\cite{katz2015necessity}. Specifying this bound for the hypothesis testing problem defined in~\eqref{eq:dht_problem} leads to 
% \begin{equation}
%  \theta \leq \sup_{\delta \in [0,1]} \min \left[   R - (H_b(p * \delta) - H_b(\delta) ) , \log(2) - H_b(p*\delta)    \right]
% \end{equation}
\begin{align}\label{eq:error_exponent_theta}
 \theta  \leq & \sup_{\delta \in [0,1]}  \min \bigg\{ R-[H_{2}(p_0*\delta)-H_{2}(\delta)],  \\
 & \left. (p_0*\delta)\log\frac{p_0*\delta}{p_1*\delta}+(1-(p_0*\delta))\log\frac{1-(p_0*\delta)}{1-(p_1*\delta)} \right\}. \notag
\end{align}
Here, $H_2$ is the binary entropy function, and $*$ is the binary convolution operator defined as $ x * y = (1-x)y + (1-y)x $, with $0 \leq x,y \leq 1$. 

\subsection{Short-length nature of DHT}
While the lower bound on the error exponent $\theta$ in~\eqref{eq:error_exponent_theta} provides a scaling law for Type-II error probability, it remains an asymptotic result due to the limit as $n$ tends to infinity in definition~\eqref{eq:def_ee}.
Nevertheless, it confirms the intuition that the problem inherently involves short sequences. For instance, consider parameters $p_0 = 0.05$, $p_1 = 0.5$, $\delta=0.1$, $R=0.4$.  When evaluating the quantity $ e^{-n\theta} $ for these parameters with $n=100$, the result is approximately $10^{-12}$, and for $n=50$, it yields approximately $10^{-6}$. This strongly suggests that practical schemes should focus on values of $n$ less than $50$.
Hence, we now introduce three practical coding schemes tailored for such short sequence lengths.

\section{Truncation Scheme}\label{sec:truncated}

\subsection{Code construction}
When considering DHT, there is no need to reconstruct all the exact values of the source bits $\mathbf{x}^n$. So a first straightforward solution consists of sending the first $\ell$ symbols of the source vector $\mathbf{x}^n$ at coding rate $R = \ell/n$. The decoder can then perform a standard Neyman-Pearson (NP) test~\cite{lehmann2005testing} on the pair $(\mathbf{x}^\ell,\mathbf{y}^\ell)$. 

Under a certain constraint $\alpha_n <  \epsilon$ on Type-I error probability, the NP lemma~\cite{lehmann2005testing} states that the following test:
\begin{equation}\label{eq:NPtest}
 \mathbb{P}_1(\mathbf{x}^\ell,\mathbf{y}^\ell) < \mu  \mathbb{P}_0(\mathbf{x}^\ell,\mathbf{y}^\ell) ,
\end{equation}
minimizes Type-II error probability $\beta_n$, where $\mu$ is a threshold value chosen to satisfy the Type-I error constraint. In~\eqref{eq:NPtest}, $\mathbb{P}_0$ and $\mathbb{P}_1$ are  the joint probability distributions of $(\mathbf{x}^\ell,\mathbf{y}^\ell)$ under hypothesis $\mathcal{H}_0$ and under hypothesis $\mathcal{H}_1$, respectively. 
Given that $p_0 < p_1$, it is shown in\cite{lehmann2005testing} that the test described by~\eqref{eq:NPtest} is equivalent to the condition:
\begin{equation}
 \sum_{i=1}^\ell (x_i \oplus y_i)  < \lambda_t,
\end{equation}
where $\lambda_t \in \mathbb{N}$ is an integer threshold value chosen so as to satisfy the constraint $\alpha_n < \epsilon$. Interestingly, this test depends on the value of the parameter $p_0$ solely through the choice of the value of $\lambda_t$. 

\subsection{Theoretical analysis}\label{sec:th_analysis_truncate}
For this scheme, analytical expressions of Type-I and Type-II error probabilities are given by~\cite{lehmann2005testing}
\begin{align}
 \alpha_n^{(t)} & = 1 - \sum_{j=0}^{\lambda_t} \binom{\ell}{j} p_0^j (1-p_0)^{\ell-j}, \\
 \beta_n^{(t)} & = \sum_{j=0}^{\lambda_t} \binom{\ell}{j}  p_1^j (1-p_1)^{\ell-j}. 
\end{align}
The truncation scheme can be seen  as a ``no-coding'' setup~\cite{gastpar2003code}. In our numerical results, it will serve as a baseline when evaluating the performance of the proposed coding schemes.  

\section{Quantization Scheme}\label{sec:quantizer}
In the literature of information theory, the seminal work of~\cite{ahlswede1986hypothesis} proposed to build a DHT coding scheme from a quantizer alone. We now introduce a practical short-length implementation of this scheme by using linear block codes. 
\subsection{Code construction}
In order to perform binary quantization, we consider the generator matrix $G_q$ of size $n\times m$ of a linear block code~\cite{richardson2008modern}. 
Then, for a given source vector $\mathbf{x}^n$ of length $n$, the encoder produces a vector $\mathbf{z}_q^m$ as~\cite{4036087}
\begin{equation}\label{eq:quantizer}
 \mathbf{z}_q^m = \arg\min_{\mathbf{z}^m} ~ d(G_q \mathbf{z}^m, \mathbf{x}^n) .
\end{equation}
%where $\text{d}_H(.,.)$ is the Hamming distance. 
The codeword $\mathbf{z}_q^m$ is transmitted to the decoder at a code rate $R=m/n$.

In~\cite{fridrich2007binary,wainwright2009low}, it is proposed to build efficient binary quantizers using low density generator matrices (LDGM). LDGM codes were considered so as to develop a low complexity message-passing algorithm called Bias-Propagation to solve~\eqref{eq:quantizer}.  However, the schemes introduced in~\cite{fridrich2007binary,wainwright2009low} consider very long codes (more than $10^5$ bits). Here, due to the short-length nature of the problem, we choose to discard the Bias-Propagation algorithm since it may lead to an important loss in performance on the considered codes. Instead, we will solve~\eqref{eq:quantizer} exactly by exhaustive search. Therefore, we consider any generator matrix $G_q$, not necessarily obtained from an LDGM code. 

The decoder first computes the quantized vector  $\mathbf{x}_q^n = G_q \mathbf{z}_q^m$. Then, since $\mathbb{P}(\mathbf{z}_q^m,\mathbf{y}^n) = \mathbb{P}(\mathbf{x}_q^n,\mathbf{y}^n)$, the NP test~\eqref{eq:NPtest} reduces to
\begin{equation}
 \sum_{i=1}^n (x_{q,i} \oplus y_i)  < \lambda_q,
\end{equation}
 where $\lambda_q$ is an integer threshold. 
 Here, compared to the truncation scheme, the decision is taken from longer vectors $\mathbf{x}_q^n$ and $\mathbf{y}^n$ of dimensions $n>m$. On the other hand, the vector $\mathbf{x}_q^n$ contains quantization errors compared to $\mathbf{x}^n$.

\subsection{Theoretical analysis}\label{sec:th_analysis_quantizer}
We now provide exact analytical expressions of Type-I and Type-II error probabilities for the quantization scheme. 
Consider the set of integers $\{E_\gamma^{(q)}\}_{\gamma \in \llbracket 0,d_{\text{max}}^{(q)}\rrbracket}$, where $E_\gamma^{(q)}$ is the number of words $\mathbf{x}^n$ of Hamming weight $\gamma$ that belong to the decision region $\mathcal{C}_0^{(q)}$ of $\mathbf{x}_q^n = \mathbf{0}^n$. 
In other words, $\mathbf{x}^n \in \mathcal{C}_0^{(q)}$ means that the solution of~\eqref{eq:quantizer} for $\mathbf{x}^n$ is $\mathbf{0}^m$. We further denote $ N_0^{(q)} = \sum_{\gamma=0}^{d_{\text{max}}^{(q)}} E_{\gamma}^{(q)}$. %We consider the codeword $\mathbf{0}^n$ because of the symmetry of the code around all of its codewords.

\begin{prop}\label{prop:error_quantization}
For the quantization scheme  and for a threshold value $\lambda_q$, Type-I and Type-II error probabilities are given by
\begin{align}
  \alpha_n^{(q)} & = 1 - \frac{1}{N_0^{(q)}} \sum_{\lambda=0}^{\lambda_q} \sum_{\gamma = 0}^{d_{\text{max}}^{(q)}} \sum_{j=0}^n E_\gamma^{(q)} \Gamma_{\lambda,j,\gamma}  p_0^j (1-p_0)^{n-j}, \label{eq:alphaq}  \\
 \beta_n^{(q)} & =  \frac{1}{N_0^{(q)}} \sum_{\lambda=0}^{\lambda_q} \sum_{\gamma = 0}^{d_{\text{max}}^{(q)}} \sum_{j=0}^n E_\gamma^{(q)} \Gamma_{\lambda,j,\gamma} p_1^j (1-p_1)^{n-j}, \label{eq:betaq}
\end{align}
where for $j = \gamma + \lambda - 2u$ and $0 \leq u \leq \min(\gamma,\lambda) \leq n$,
\begin{equation}\label{eq:Gamma_def}
 \Gamma_{\lambda,j,\gamma} = \binom{\gamma}{u} \binom{n-\gamma}{\lambda - u} . 
\end{equation}
\end{prop}
\begin{proof}
%In what follows, we simplify the notation by using $\mathbb{P}_0$ (respectively $\mathbb{P}_1$) to denote a probability distribution under hypothesis $\mathcal{H}_0$ (respectively under hypothesis $\mathcal{H}_1$). 
Since by symmetry the quantizer error probability is independent of the transmitted codeword~\cite{richardson2008modern}, we consider the all-zero codeword $\mathbf{x}_q^n = \mathbf{0}$.  
 We develop
 \begin{align}\label{eq:proof_q1}
   \alpha_n^{(q)} & = 1 - \sum_{\lambda=0}^{\lambda_q} \mathbb{P}_0(w(\mathbf{Y}^n)  = \lambda) \\ \notag
   & = 1 -   \sum_{\lambda=0}^{\lambda_q}\sum_{\gamma=0}^{d_{\text{max}}^{(q)}} \frac{E_\gamma^{(q)}}{N_0^{(q)}} \mathbb{P}_0(w(\mathbf{Y}^n)=\lambda | w(\mathbf{X}^n) = \gamma) \\ \notag
   & = 1 -   \sum_{\lambda=0}^{\lambda_q}\sum_{\gamma=0}^{d_{\text{max}}^{(q)}} \frac{E_\gamma^{(q)}}{N_0^{(q)}} \sum_{j=0}^n \mathbb{P}_0(d(\mathbf{X}^n,\mathbf{Y}^n) = j )  \frac{\Gamma_{\lambda,j,\gamma}}{ \binom{n}{j}} \\ \notag
   & = 1 - \sum_{\lambda=0}^{\lambda_q}\sum_{\gamma=0}^{d_{\text{max}}^{(q)}} \frac{E_\gamma^{(q)}}{N_0^{(q)}} \sum_{j=0}^n \binom{n}{j} p_0^j (1-p_0)^{n-j} \frac{\Gamma_{\lambda,j,\gamma}}{ \binom{n}{j}} .
 \end{align}
This gives~\eqref{eq:alphaq}. 
To obtain~\eqref{eq:betaq}, we remark that $ \beta_n^{(q)} = \sum_{\lambda=0}^{\lambda_q} \mathbb{P}_1(w(\mathbf{Y}^n)  = \lambda) $ and follow the same steps as in~\eqref{eq:proof_q1}, by replacing $p_0$ by $p_1$. 
\end{proof}
%Note that the term~\eqref{eq:Gamma_def} already appeared in~\cite{polyanskiy2010channel} in the expression of the lower bound on the decoder error probability after transmission on a binary symmetric channel (BSC). 

\section{Quantize-Binning Scheme}\label{sec:qandb}
We now propose a practical solution for the quantize-binning scheme which has been widely investigated in the literature of information theory for DHT~\cite{shimokawa1994error,rahman2012optimality,katz2015necessity}.
\subsection{Code construction}
In the quantize-binning scheme, we consider as before a generator matrix $G_q$ of size $n\times m$. We also resort to the parity check matrix $H_b$ of size $k \times m$ of another linear block code. After using $G_q$ for binary quantization as described in Section~\ref{sec:quantizer}, the encoder uses the matrix $H_b$ to compute
\begin{equation}
 \mathbf{u}^k = H_b \mathbf{z}_q^m . 
\end{equation}
The syndrom $ \mathbf{u}^k$ is then transmitted to the decoder.  In this case, the coding rate is given by $R = k/n$. 
At the decoder, in order to apply the NP test~\eqref{eq:NPtest}, we first identify by exhaustive search a vector $\hat{\mathbf{z}}_q^m$ as
\begin{equation}\label{eq:decoder_qandb}
 \hat{\mathbf{z}}_q^m = \arg\min_{\mathbf{z}^m} ~ d(G_q\mathbf{z}^m,\mathbf{y}^n) \text{ s.t. } H_b\mathbf{z}^m = \mathbf{u}^k .
\end{equation}
 We then apply the following test: 
\begin{equation}
 \sum_{i=1}^n (\hat{x}_{q,i} \oplus y_i)  < \lambda_{qb},
\end{equation}
where $\hat{\mathbf{x}}_q^n = G_q \hat{\mathbf{z}}_q^m$, and  $\lambda_{qb}$ is an integer threshold.   

Next, according to~\cite{katz2015necessity}, the binning allows us to leverage the side information vector $\mathbf{y}^n$ so as to further reduce the coding rate. However, it also introduces a binning error probability which can impact Type-I and Type-II error probabilities. While this problem was well investigated in the asymptotic regime~\cite{katz2015necessity}, we next discuss it for the considered short length code construction. 

% To reduce the decoder complexity without resorting to a sub-optimal algorithm nor to an exhaustive approach,  %we propose the following method in order to avoid testing the $2^m$ possible vectors $\mathbf{z}^m$ in order to solve~\eqref{eq:decoder_qandb}. 
% we  create an initial vector $\tilde{\mathbf{z}}^m$ that contains the first $m$ components of $\mathbf{y}^n$. We then test condition~\eqref{eq:decoder_qandb} onto all vectors $\mathbf{z}^m$ such that $\text{d}_H(\tilde{\mathbf{z}}^m,\mathbf{z}^m) \leq d_B$, where $d_B$ is a parameter of the algorithm. This amounts to testing $\sum_{d=0}^{d_B} \binom{m}{d} $ candidate vectors $\mathbf{z}^m$ instead of $2^m$ candidates. 

\subsection{Theoretical analysis}\label{sec:th_analysis_qb}
We now consider the decision region $\mathcal{C}_0^{(qb)}$ for the all-zero codeword of the quantize-binning scheme. Especially, a side information vector $\mathbf{y}^n$ belongs to $ \mathcal{C}_0^{(qb)}$ if the solution of~\eqref{eq:decoder_qandb} for this vector is $\hat{\mathbf{z}}_q^m = \mathbf{0}^m$. We then define the set of integers $\{E_\nu^{(qb)}\}_{\nu \in \llbracket 0,d_{\text{max}}^{(qb)}\rrbracket}$, where $E_\nu^{(qb)}$ is the number of words $\mathbf{y}^n$ of Hamming weight $\nu$ that belong to the decision region $\mathcal{C}_0^{(qb)}$.  
We also define the set of integers $\{A_t^{(qb)}\}_{t \in \llbracket 0,n \rrbracket }$, where $A_t^{(qb)}$ is the number of codewords $\mathbf{x}_q^n$ of Hamming weight~$t$ such that there exists $\mathbf{z}_q^m$ that satisfies $\mathbf{x}_q^n = G_q \mathbf{z}_q^n$, and $H_b \mathbf{z}_q^m = \mathbf{0}^k$. As a result, the set  $\{A_t^{(qb)}\}_{t \in \llbracket 0,n \rrbracket }$ is the code weight distribution of the concatenated code. 

\begin{prop}
 For the quantize-binning scheme and for a threshold value $\lambda_{qb}$, Type-I and Type-II error probabilities are given by
 \begin{align}
  \alpha_n^{(qb)} & = 1 - \mathbb{P}_{B}(p_0) -  \mathbb{P}_{\bar{B}}(p_0),\\
  \beta_n^{(qb)} & = \mathbb{P}_{B}(p_1) +  \mathbb{P}_{\bar{B}}(p_1) ,
 \end{align}
 where
 \begin{align}\label{eq:Pb0Blambda}
  \mathbb{P}_{B}(\delta)  & = \sum_{\nu=0}^{\min(d_{\max}^{(qb)},\lambda_{qb})} \frac{E_{\nu}^{(qb)}}{\binom{n}{\nu}} \sum_{\gamma=0}^{d_{\max}^{(q)}} \frac{E_\gamma^{(q)}}{N_0^{(q)}} \sum_{j=0}^n \Gamma_{\nu,j,\gamma} \delta^j (1-\delta)^{n-j} ,  \\   \label{eq:Pb0barBlambda}
  \mathbb{P}_{\bar{B}}(\delta) & = \sum_{i=0}^n \left[ \left( \sum_{\gamma=0}^{d_{\text{max}}^{(q)}} \frac{E_{\gamma}^{(q)}}{N_0^{(q)}} \sum_{j=0}^n \Gamma_{i,j,w} \delta^j (1-\delta)^{n-j} \right) \right. \\ \notag
  & ~~~~~~~~~~~~~~~~~~~~~~ \left. \times \left( \sum_{t=1}^n \sum_{\nu=0}^{\lambda_{qb}} \frac{E_{\nu}^{(qb)}}{\binom{n}{\nu}} \frac{A_t^{(qb)} \Gamma_{i,\nu,t}}{\binom{n}{i}}  \right) \right] .
\end{align}
\end{prop}
\begin{proof}
We consider the all-zero codeword $\mathbf{x}_q^n = \mathbf{0}$. 
 %We use $\bar{B}$ to denote the event that an incorrect sequence was extracted from the bin of index $\mathbf{u}^k = \mathbf{0}$. 
 Under the hypothesis $\mathcal{H}_0$, we express
 \begin{equation}
  \alpha_n^{(qb)} = 1 -  \mathbb{P}_0(\mathcal{\widehat{H}}_0,B) -  \mathbb{P}_0(\mathcal{\widehat{H}}_0,\bar{B}) .
 \end{equation}
In this expression, $B$ is the event that the correct sequence $\hat{\mathbf{x}}_q = \mathbf{x}_q$ was retrieved at the decoder, while $\bar{B}$ is the event that an incorrect sequence $\mathbf{\hat{x}}_q \neq \mathbf{x}_q $ was output by the decoder. In addition, $\mathcal{\widehat{H}}_0$ is the event that hypothesis $\mathcal{H}_0$ was decided at the decoder. We further denote  $ \mathbb{P}_{B}(p_0) =  \mathbb{P}_0(\mathcal{\widehat{H}}_0,B)$ and $ \mathbb{P}_{\bar{B}}(p_0) =  \mathbb{P}_0(\mathcal{\widehat{H}}_0,\bar{B})$. 
We then express
\begin{align}
\mathbb{P}_{B}(p_0) &=  \sum_{\nu=0}^n \mathbb{P}_0(w(\mathbf{Y}^n)=\nu) \mathbb{P}_0(\mathcal{\widehat{H}}_0,B|w(\mathbf{Y}^n)=\nu) \\ %\mathbb{P}_0(\mathcal{\widehat{H}}_0|B,w(\mathbf{y}^n)=i) . 
& = \sum_{\nu=0}^{\min(d_{\max}^{(qb)},\lambda_{qb})}  \mathbb{P}_0(w(\mathbf{Y}^n)=\nu) \frac{E_{\nu}^{(qb)}}{\binom{n}{\nu}} .
\end{align}
% \begin{equation}
% \mathbb{P}_{B}(p_0) = \sum_{i=0}^n \mathbb{P}_0(w(\mathbf{y}^n)=i) \mathbb{P}_0(B|w(\mathbf{y}^n)=i) \mathbb{P}_0(\mathcal{\widehat{H}}_0|B,w(\mathbf{y}^n)=i) . 
% \end{equation}
%In this expression, $\mathbb{P}_0(B|w_\mathbf{y}=\gamma) = \frac{E_{\gamma}^{(qb)}}{\binom{n}{\gamma}}$ if $\gamma \leq d_{\text{max}}^{(qb)}$, and $\mathbb{P}_0(B|w_\mathbf{y}=w) = 0$ otherwise. In addition, $\mathbb{P}_0(\mathcal{H}_0|B,w_\mathbf{y}=\gamma) = 1$ if $\gamma \leq \lambda$, and $\mathbb{P}_0(\mathcal{H}_0|B,w_\mathbf{y}=\gamma) = 0$ otherwise. 
Next, by following the same steps as in the proof of Proposition~\ref{prop:error_quantization}, we show that 
\begin{equation}\label{eq:Pwy}
  \mathbb{P}_0(w(\mathbf{Y}^n)=\nu) =  \sum_{\gamma=0}^{d_{\max}^{(q)}} \frac{E_\gamma^{(q)}}{N_0^{(q)}} \sum_{j=0}^n \Gamma_{\nu,j,\gamma} p_0^j (1-p_0)^{n-j} ,
\end{equation}
which provides~\eqref{eq:Pb0Blambda}. 
We then write 
\begin{equation}
  \mathbb{P}_{\bar{B}}(p_0) = \sum_{i=0}^n \mathbb{P}_0(w(\mathbf{Y}^n)=i) \mathbb{P}_0(\mathcal{\widehat{H}}_0,\bar{B}|w(\mathbf{Y}^n)=i),
\end{equation}
where $\mathbb{P}_0(w(\mathbf{Y}^n)=i)$ is given by~\eqref{eq:Pwy}. Next, we develop
\begin{align}\notag
 & \mathbb{P}_0(\mathcal{\widehat{H}}_0,\bar{B}|w(\mathbf{Y}^n)=i)  \\ & = \sum_{t=1}^n \sum_{\nu=0}^{\lambda_{qb}} \mathbb{P}_0(w(\hat{\mathbf{X}}_q^n) = t,d(\hat{\mathbf{X}}_q^n,\mathbf{Y}^n) = \nu |w(\mathbf{Y}^n) = i) \\
 & = \sum_{t=1}^n \sum_{\nu=0}^{\lambda_{qb}} \frac{E_\nu^{(qb)}}{\binom{n}{\nu}}  \frac{A_t^{(qb)} \Gamma_{i,\nu,t}}{\binom{n}{i}}.
\end{align}
% \begin{align}\notag
% &  \mathbb{P}_0(\mathcal{\widehat{H}}_0,\bar{B}|w(\mathbf{y}^n)=i) = \sum_{t=1}^n \sum_{\eta=0}^{\lambda_{qb}} \mathbb{P}_0(w_{\hat{x}_q} = t,d_{\hat{x}_q,\mathbf{y}} = k |w_{\mathbf{y}} = \gamma) \\
% & ~~~~~~~~~~~~ \times \mathbb{P}_0(\mathcal{H}_0|\bar{B},w_{\mathbf{y}}=\gamma,w_{\hat{\mathbf{x}}_q} = t, d_{\hat{x}_q,\mathbf{y}} = k )
% \end{align}
% since $\mathbb{P}_0(\bar{B} | w_{\mathbf{y}}=\gamma,w_{\hat{\mathbf{x}}_q} = t, d_{\hat{x}_q,\mathbf{y}} = k ) = 1$ for $t > 0$ and $ \mathbb{P}_0(\mathcal{H}_0|\bar{B},d_{\hat{\mathbf{x}}_q,\mathbf{y}} = k ) = 1$ if $k\leq \lambda_{qb}$, and $0$ otherwise. Finally,
% \begin{equation}
%   \mathbb{P}_0(w_{\hat{x}_q} = t,d_{\hat{x}_q,\mathbf{y}} = k |w_{\mathbf{y}} = \gamma) = \frac{E_k^{(qb)}}{\binom{n}{k}} \frac{A_t \Gamma_{\gamma,k,t}}{\binom{n}{\gamma}} ,
% \end{equation}
This provides the expression of $ P_{\bar{B}}(p_0)$ in~\eqref{eq:Pb0barBlambda}.

We obtain the expression of $\beta_n^{(qb)}$ from the previous equations by noticing that  $\beta_n^{(qb)} = \mathbb{P}_{B}(p_1) + \mathbb{P}_{\bar{B}}(p_1)$.   
 
 \end{proof}
\section{Numerical Results}\label{sec:simulation}
This section provides Monte-Carlo simulation results for the proposed code constructions, and compares them with the theoretical Type-I and Type-II error probabilities. 

% \textcolor{red}{Why BCH codes?}

\subsection{Truncation versus quantization}
First of all, we evaluate the performance of the quantization scheme introduced in Section~\ref{sec:quantizer} and compare it against the truncation scheme described in Section~\ref{sec:truncated}. We set parameters $p_1 = 0.5$ (testing against independence), and $p_0 = 0.1$ or $p_0 = 0.07$.  For the quantization, we consider the BCH $(31,16)$-code with minimum distance $d_{\min} = 7$.  As a result, after applying the quantizer, $m=16$ bits are sent to the decoder. Therefore, we consider for comparison the truncation scheme with $\ell=16$. 
For these two schemes, Fig.~\ref{fig:Quantization_BCH3116} shows the receiver operating characteristic (ROC) curves which provide the Type-II error with respect to the Type-I error, obtained for different threshold values $\lambda_t, \lambda_q \in \llbracket 1,m \rrbracket$. Note that the ROC curves are usually considered for the evaluation of hypothesis tests. In Fig.~\ref{fig:Quantization_BCH3116}, the plain curves come from Monte-Carlo simulations averaged other $10000$ trials, while the dashed curves come from the error probability expressions provided in Sections~\ref{sec:th_analysis_truncate} and~\ref{sec:th_analysis_quantizer}. 

For the considered two values of $p_0$, we observe that the quantizer scheme performs better than the truncation scheme. This is due to the fact that with the quantizer, the decision is made on $n=31$ bits instead of $m=16$ bits, although the quantized vector $\mathbf{x}_q^n$ contains errors compared to the original $\mathbf{x}^n$. In addition, we observe that the theoretical Type-I and Type-II error probabilities are closely consistent with the Monte Carlo results. This is because the error probability expressions take into account the considered code through the terms $E_\gamma^{(q)}$. As a result, the theoretical expressions are found to be relevant tools for the DHT code design.   

 \fig{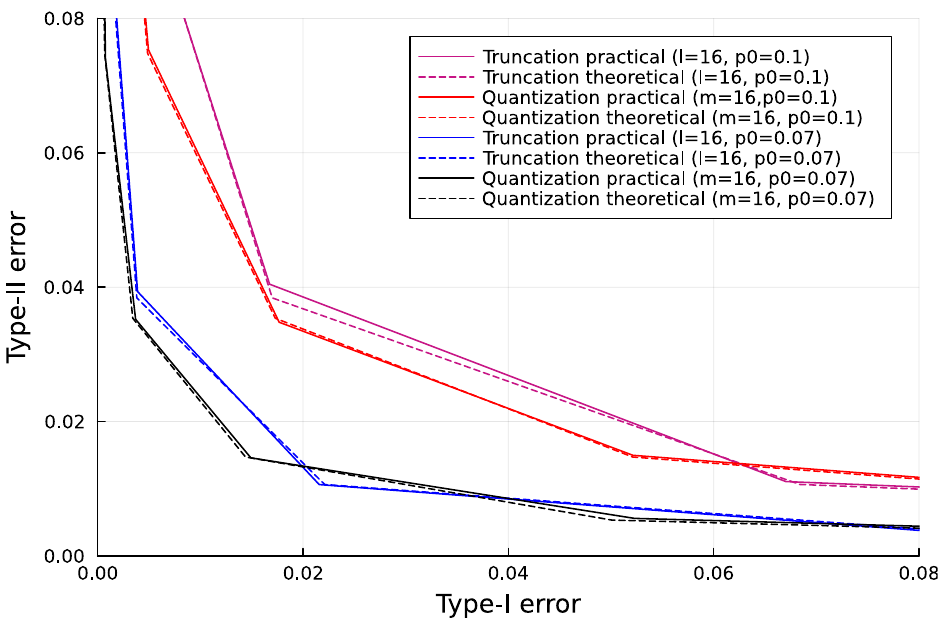}{ROC curve for the BCH code $(31,16,7)$ user as a quantizer, compared to the truncation scheme.}{0.4}{t}

 \subsection{Truncation versus quantize-binning}
 We now evaluate the performance of the quantize-binning scheme proposed in Section~\ref{sec:qandb} compared to the truncation scheme. For the quantize-binning scheme, we consider the BCH $(31,16)$ code for the quantizer, and the Reed-Muller $(16,5)$ code with $d_{\text{min}}=8$ for the binning part. As a result, only $k=11$ coded bits are sent to the decoder. Therefore, for comparison, we consider the truncation scheme with $\ell=11$.  
 Fig.~\ref{fig:QandB} shows the ROC curves for the considered two schemes, both for $p_1=0.35$ and $p_0 = 0.01$ or $p_0 = 0.03$. In this case again, plain curves come from Monte-Carlo simulations averaged over $10000$ trials, and dashed curves come from the error probability expressions obtained in Sections~\ref{sec:th_analysis_truncate} and~\ref{sec:th_analysis_qb}.  We observe that the quantize-binning scheme performs much better than the truncation scheme, since it makes decision on $n=31$ bits instead of $\ell=11$ bits for the truncation scheme, despite the fact that both quantizer and binning can introduce errors. We also observe that the theoretical Type-I and Type-II error probabilities are closely consistent with the practical performance.
 
 Finally, it is worth noting that the performance of the quantize-binning scheme strongly depends on the considered pair of codes used for quantization and binning. In this regard, the derived theoretical error probability expressions can serve as a useful tool for the design of efficient pairs of codes. 
Also we did not include the entropy-check condition of~\cite{katz2015necessity} in our quantize-and-binning scheme, since this did not bring any further improvement in terms of practical Type-I and Type-II error probabilities compared to the truncated scheme.

%  Finally, note that the performance of the quantize-binning scheme strongly depends on the considered pair of codes for quantization and binning. Especially, the considered concatenated BCH with Reed-Muller code has a minimum distance $d_{\min} = 11$, which achieves a good final performance. However, we have also tested some other combinations which led to equivalent or worse performance compared to the truncated scheme. In this regard, it is worth noting that the proposed theoretical error probability expression can serve as a useful tool for identifying efficient code designs. 
%  
  \fig{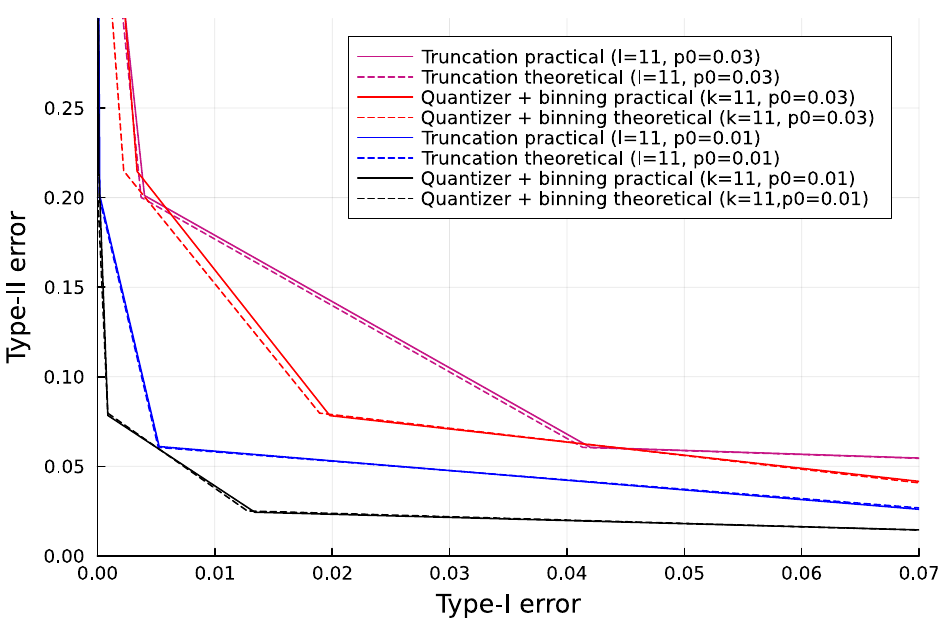}{ROC curve for the quantize-binning scheme built from the BCH code $(31,16,7)$ for quantization combined with the Reed-Muller code $(16,5,8)$ for binning}{0.4}{t}

\section{Conclusion}
In this paper, we have introduced two practical coding schemes for binary DHT, one built with a binary quantizer, and the other built with a quantize-binning scheme. Both schemes were designed from short linear block codes. For each considered scheme, we also derived theoretical expressions of Type-I and Type-II error probabilities. Simulation results demonstrated the superiority of the proposed schemes compared to the baseline truncation scheme, and also showed the accuracy of the proposed theoretical expressions. From a practical point of view, future works will include an interleaver design to improve the performance of the concatenated construction,  as well as complexity reduction of the decoders so as to allow for larger code length to be considered. %Another topic left as future work is to identify the relationship between the theoretical error probabilities derived in this paper and usual code performance bounds such as Hamming and Gilbert Varshamov bounds. 

\bibliography{sample}
\bibliographystyle{IEEEtran}

\end{document}